\documentclass[12pt]{article}
\usepackage{amsmath,amssymb,dsfont}
\usepackage{graphicx}

\newtheorem{theorem}{Theorem}

\newtheorem{corollary}[theorem]{Corollary}

\newtheorem{definition}[theorem]{Definition}
\newtheorem{example}[theorem]{Example}

\newtheorem{lemma}[theorem]{Lemma}

\newtheorem{remark}[theorem]{Remark}

\newenvironment{proof}[1][Proof]{\noindent\textbf{#1.} }{\ }
\newenvironment{keywords}[1][Keywords]{\noindent\textbf{#1:} }{}

\def\Rbb{\mathbb{R}}
\def\Cbb{\mathbb{C}}

\def\half{\frac{1}{2}}
\def\quarter{\frac{1}{4}}
\def\diag{\hbox{diag}}
\def\isym{\mathfrak{i}}

\begin{document}

\title{On synthesis of linear quantum stochastic systems
by pure cascading\footnote{This research is supported by the
Australian Research Council.}}
\author{Hendra~I.~Nurdin 
\thanks{H.~I.~Nurdin is with the Department of Information Engineering,
Research School of Information Sciences and Engineering (Building
115), The Australian National University, Canberra ACT 0200,
Australia. Phone: +61-2-61258656 Fax: +61-2-61258660 Email:
Hendra.Nurdin@anu.edu.au.}}

\maketitle 
\begin{abstract}
Recently, it has been demonstrated that
an arbitrary linear quantum stochastic system can be realized as a
cascade connection of simple one degree of freedom quantum
harmonic oscillators together with a direct interaction
Hamiltonian which is bilinear in the canonical operators of the
oscillators. However, from an experimental point of view,
realizations by pure cascading, without a direct interaction
Hamiltonian, would be much simpler to implement and this raises
the natural question of what class of linear quantum stochastic
systems are realizable by cascading alone. This paper gives a
precise characterization of this class of linear quantum stochastic systems 
and then it is proved that, in the {\em
weaker} sense of transfer function realizability, {\em all} passive
linear quantum stochastic systems belong to this class. A
constructive example is given to show the transfer function
realization of a two degrees of freedom passive linear quantum
stochastic system by pure cascading.
\end{abstract}


\section{Background and Motivation}
Recently, there has been interest in the literature on control of
a linear quantum stochastic system with a controller which is a
quantum system of the same type \cite{YK03b,JNP06,NJP07b,Mab08},
often referred to as ``coherent-feedback control''.
The potential applications for linear quantum stochastic systems include quantum information processing and photonic signal processing. For instance, they can act as the coherent photonic circuitry subsystem in a cavity  QED  system, the latter system being realized by placing suitable atoms inside the optical  cavities in a linear quantum stochastic system. Cavity QED networks are of interest for quantum information processing (see, e.g., \cite{NC00}), such as in the quantum internet \cite{Kimb08}, 
whilst the controller realized in \cite{Mab08} is an early sample application of linear quantum stochastic systems to photonic signal processing.  

The  studies on coherent-feedback control naturally led to the consideration of the network synthesis problem for linear quantum stochastic systems \cite{NJD08}, which may be viewed as a quantum analogue of the network synthesis problem for linear electrical systems \cite{AV73}. Nurdin, James and Doherty \cite{NJD08} have shown that any linear quantum
stochastic system can, in principle, be synthesized by a cascade
of simple one degree of freedom harmonic oscillators together with
a direct interaction Hamiltonian between the canonical operators
of these oscillators.  Alternative schemes have subsequently been proposed in \cite{Nurd09b,Pet09}, but we note that \cite{Pet09} considers a weaker type of realizability than in \cite{NJD08,Nurd09b},  i.e., transfer function realizability (cf. section \ref{sec:synthesis-problem}), and the results therein limited to a certain sub-class of linear quantum stochastic systems.

From an experimental perspective, direct bilinear
interaction Hamiltonians between independent harmonic oscillators
are challenging to implement for systems that have more than just
a few degrees of freedom and therefore it becomes important to
investigate what kind of systems can be realized by a pure cascade
connection. A key result of this paper is a necessary and sufficient condition  for a linear quantum
stochastic system to be realizable by only a cascade connection of one degree of freedom oscillators, without any direct interaction Hamiltonian. Moreover, we also show that the associated transfer functions of {\em all}
passive linear quantum stochastic systems can always be realized by a cascade connection,
proving in general the partial results of \cite{Pet09} without the additional assumptions made therein.

The organization of this paper is as follows. Section \ref{sec:prelim} sets up the notations and gives a brief overview
of linear quantum stochastic systems. Section \ref{sec:synthesis-problem} defines the synthesis
problem and discusses the
notions of strict realizability and transfer function
realizability. Section \ref{sec:cascade-conds} derives  a necessary and sufficient condition for a linear
quantum system to be realizable by a pure cascade connection of
one degree of freedom quantum harmonic oscillators, in both the
strict and transfer function sense of realizability. Section \ref{sec:passive-sys}
then introduces the class of passive linear quantum systems and
proves that all such systems are transfer functions realizable by
a pure cascade connection. Finally, section \ref{sec:conclusions} offers some
conclusions of this paper. 

\section{Preliminaries}
\label{sec:prelim} 

\subsection{Notation}
We shall use the following notations: $\isym=\sqrt{-1}$, $^*$ denotes the adjoint of a linear operator
as well as the conjugate of a complex number. If $A=[a_{jk}]$ then $A^{\#}=[a_{jk}^*]$, and $A^{\dag}=(A^{\#})^T$, where $^T$ denotes matrix transposition.  $\Re\{A\}=(A+A^{\#})/2$ and $\Im\{A\}=\frac{1}{2\isym}(A-A^{\#})$,
and denote the identity matrix by $I$ whenever its size can be
inferred from context and use $I_{n}$ to denote an $n \times n$
identity matrix. Similarly, $0$ denotes  a matrix with zero
entries whose dimensions can be determined from context. 
$\diag(M_1,M_2,\ldots,M_n)$  denotes a block diagonal matrix with
square matrices $M_1,M_2,\ldots,M_n$ on its diagonal block, and $\diag_{n}(M)$
a block diagonal matrix with the
square matrix $M$ appearing on its diagonal blocks $n$ times. 

\subsection{The class of linear quantum stochastic systems}
\label{sec:LQSS-prelim}
In this paper, we will be concerned with a class of quantum stochastic models of open (i.e., quantum systems  that can interact with an environment) {\em Markov} quantum systems  that are widely used and are standard in quantum optics. Such models have been in the physics and mathematical physics literature since the 1980's, see, e.g., \cite{HP84,GC85,KRP92,GZ04,WM10}. In particular, we focus on the special sub-class of {\em linear} quantum stochastic models, see, e.g.,  \cite[section 6.6]{WM10}, \cite[sections 3, 3.4.3, 5.3, chapters 7 and 10]{GZ04}, \cite[section 4]{EB05}, \cite[section 5]{BE08}, \cite{WD05,JNP06,NJP07b,NJD08,Yam06,Mab08,YNJP08,GGY08}.
These linear quantum stochastic models describe such quantum optical devices as optical cavities \cite[section 5.3.6]{BR04}\cite[chapter 7]{WaM94}, linear quantum amplifiers \cite[chapter 7]{GZ04}, and finite bandwidth squeezers \cite[chapter 10]{GZ04}. Following the terminology in \cite{JNP06,NJP07b,NJD08}, we shall refer to this class of models as {\em linear quantum stochastic systems}.

Suppose we have $n$ independent quantum harmonic oscillators labelled $1,\ldots,n$. Each oscillator $j$ has position and momentum operators $q_j$ and $p_j$, respectively. The position and momentum operators satisfy the canonical commutation relations $[q_j,p_k]=2\isym \delta_{jk}$, $[q_j,q_k]=0$, and $[p_j,p_k]=0$, where $\delta_{jk}$ denotes the Kronecker delta that takes on the value 1 only if $j=k$, but is otherwise 0. Equivalently, we may describe them in terms of the $2n$ annihilation and creation operators $a_1,a_1^*,a_2,a_2^*,\ldots,a_n,a_n^*$, with $a_j=(q_j+\isym p_j)/2$, satisfying the canonical commutation relations $[a_j,a^*_k]=  \delta_{jk}$, $[a_j,a_k]=0$ and $[a^*_j,a^*_k]=0$.  
The independent oscillators can be coupled to one or more external independent quantum fields, say $m$ of them.  In a Markov quantum system, the $m$ independent fields are essentially {\em quantum noises}  modelled by bosonic annihilation field operators $\mathcal{A}_1(t), \mathcal{A}_2(t),\ldots,\mathcal{A}_m(t)$ that can be defined on a separate Fock space (over $L^2(\Rbb)$) for each field operator  \cite{HP84,KRP92,BvHJ07}. For each $\mathcal{A}_j(t)$ there is a corresponding creation field operator $\mathcal{A}_j^*(t)$ that is defined on the same Fock space and is the operator adjoint of $\mathcal{A}_j(t)$, i.e., $\mathcal{A}_j^*(t)=\mathcal{A}_j(t)^*$. The field operators are adapted quantum stochastic processes with forward differentials $d\mathcal{A}_j(t)=\mathcal{A}_j(t+dt)-\mathcal{A}_j(t)$ and $d\mathcal{A}_j^*(t)=\mathcal{A}_{j}^*(t+dt)-\mathcal{A}_j^*(t)$ that have the quantum It\^{o} products \cite{HP84,KRP92,BvHJ07}:
\begin{align*}
&d\mathcal{A}_{j}(t)d\mathcal{A}_{k}(t)^*=\delta_{jk}dt;\, d\mathcal{A}_{j}^*(t)d\mathcal{A}_{k}(t)= 0;\,d\mathcal{A}_{j}(t)d\mathcal{A}_{k}(t)=0;\\
&d\mathcal{A}_{j}^*(t) d\mathcal{A}_{k}^*(t)=0;\,d\mathcal{A}_{k}(t)dt=0;\,d\mathcal{A}_{k}^*(t)dt=0.
\end{align*}
More informally, as in the quantum Langevin formalism,  we can express $\mathcal{A}_j(t) = \int_0^t \eta_j(s)ds$ and $\mathcal{A}^*_j(t) = \int_0^t \eta^*_j(s)ds$, where $\eta_j(t)$ for $j=1,\ldots,m$ are independent quantum white noise processes satisfying the informal commutation relations $[\eta_j(s),\eta_k^*(t)]=\delta_{jk} \delta(t-s)$ and  $[\eta_j(s),\eta_k(t)]=[\eta_j(s)^*,\eta_k^*(t)]=0$, where $\eta_j^*(t)=\eta_j(t)^*$, and $\delta(t)$ denotes the Dirac delta function.
 
Let us collect the position and momentum operators in the column vector $x$ defined as $x=(q_1,p_1,q_2,p_2,\ldots,q_n,p_n)^T$. Note that in terms of $x$, we may write the canonical commutation relations as $xx^T-(xx^T)^T=2\isym\Theta$ with $\Theta=\diag_{n}(J)$. We take the composite system of $n$ quantum harmonic oscillators to have a {\em quadratic Hamiltonian} $H$ given by $H=\frac{1}{2} x^T R x$, where $R$ is a real $2n \times 2n$ symmetric matrix. The oscillators are coupled to the quantum field $m$ via the informal singular interaction Hamiltonian $H_m=\isym (L_m \eta_m^*(t)-L_m^*\eta_m(t))$ \cite{GC85,GZ04}, where $L_m=K_m x$ with $K_m \in \Cbb^{1 \times 2n}$ is a linear coupling operator of the oscillator position and momentum operators to $\eta_m(t)$. Collect the coupling operators $L_1,L_2,\ldots,L_m$ together in one {\em linear coupling vector} $L= (L_1,L_2,\ldots,L_m)^T=K x$, with $K=[\begin{array}{cccc} K_1^T & K_2^T & \ldots & K_m^T\end{array}]^T$, and the field operators together as $\mathcal{A}(t)=(\mathcal{A}_1(t),\mathcal{A}_2(t),\ldots,\mathcal{A}_m(t))^T$. Then 
the {\em joint} evolution of the oscillators and the quantum fields is given by a unitary adapted process $U(t)$ satisfying the Hudson-Parthasarathy quantum stochastic differential equation (QSDE) \cite{HP84,KRP92,BvHJ07,GJ07}:
\begin{eqnarray*}
dU(t) = ({\rm tr}((S-I)^T d\Lambda(t)) +  d\mathcal{A}(t)^{\dag} L - L^{\dag}Sd\mathcal{A}(t) -(\isym H + \frac{1}{2}L^{\dag}L dt))U(t),
\end{eqnarray*}
where $S \in \Cbb^{m \times m}$ is a complex unitary matrix (i.e., $S^{\dag}S=SS^{\dag}=I$) called the {\em scattering matrix}, and $\Lambda(t)=[\Lambda_{jk}(t)]_{j,k=1,\ldots,m}$. The processes $\Lambda_{jk}(t)$ for $j,k=1,\ldots,m$ are adapted quantum stochastic processes that are referred to as {\em gauge processes}, and the forward differentials $d\Lambda_{jk}(t)=\Lambda_{jk}(t+dt)-\Lambda_{jk}(t) $ $j,k=1,\ldots,m$ have the quantum It\^{o} products:
$$
d\Lambda_{jk}(t)
d\Lambda_{j'k'}(t)\hspace*{-1pt}=\hspace*{-1pt}\delta_{kj'}d\Lambda_{jk'}(t),\hspace*{6pt}d\mathcal{A}_j(t)
d\Lambda_{kl}(t)\hspace*{-1pt}=\hspace*{-1pt}\delta_{jk}d\mathcal{A}_l(t),\hspace*{6pt}d\Lambda_{jk} d\mathcal{A}_l(t)
^*\hspace*{-1pt}=\hspace*{-1pt}\delta_{kl}d\mathcal{A}_j^*(t),
$$
with all other remaining cross products between $d\Lambda_{jk}(t)$ and either of $dt$, $d\mathcal{A}_{j'}(t)$ or $d\mathcal{A}^*_{k'}(t)$ being zero. Informally, we may express $\Lambda_{jk}(t)=\int_{0}^t \eta^*_j(s)\eta_k(s) ds$.

For any adapted processes $V(t)$ and $W(t)$ satisfying a quantum Ito stochastic differential equation, we have the {\em quantum Ito rule}  $d(V(t)W(t))=V(t)dW(t)+(dV(t)) W(t) + dV(t) dW(t)$.  Using the quantum Ito rule and the quantum Ito products given above, as well as exploiting the canonical commutation relations between the operators in $x$, the {\em Heisenberg evolution} $X(t)=U(t)^* x U(t)$ of the canonical operators in the vector $x$ satisfies the quantum stochastic differential equation, see \cite[section 4]{EB05}, \cite[section 5]{BE08}, \cite{JNP06,NJD08}:
\begin{align}
dX(t)&= d(U(t)^* x U(t)) = \tilde AX(t)dt+\tilde B\left[\begin{array}{c} d\mathcal{A}(t)
\\ d\mathcal{A}(t)^{\#} \end{array}\right];  X(0)=x, \notag\\
dY(t)&= d(U(t)^* \mathcal{A}(t)U(t)) = \tilde C x(t)dt+ \tilde Dd\mathcal{A}(t), \label{eq:qsde-out}
\end{align}
with $\tilde A=2\Theta(R+\Im\{K^{\dag}K\})$, $\tilde B=2\isym \Theta [\begin{array}{cc}
-K^{\dag}S & K^TS^{\#}\end{array}]$,
$\tilde C=K$, and $\tilde D=S$, where  $Y(t)=(Y_1(t),\ldots,Y_m(t))^T=U(t)^* \mathcal{A}(t) U(t)$ is a vector of
{\em output fields} that results from the interaction of the quantum harmonic oscillators and the incoming quantum fields   $\mathcal{A}(t)$. Note that the dynamics of $X(t)$ is linear, while $Y(t)$ depends linearly on $X(t)$ and $\mathcal{A}(t)$. We refer to $n$ as the {\em degrees of freedom} of the oscillators. If $n=1$, we shall often refer to
the linear quantum stochastic system as a one degree of freedom (open quantum harmonic) oscillator. 

Following \cite{GJ07}, we denote a  linear quantum stochastic system with  Hamiltonian $H$,  coupling vector $L$ and scattering matrix $S$ simply as $G=(S,L,H)$ or $G=(S,Kx,\half x^TRx)$. We also recall  the {\em concatenation product} $\boxplus$ and {\em series product} $\triangleleft$  for open Markov quantum systems \cite{GJ07} defined by $G_1 \boxplus G_2=(\diag(S_1,S_2),(L_1^T,L_2^T)^T,H_1+H_2)$, and $G_2 \triangleleft G_1=(S_2S_1,L_2+S_2L_1
,H_1+H_2+\Im\{L_2^{\dag}S_2L_1\})$. Since both products are associative, the products $G_1
\boxplus G_2 \boxplus \ldots \boxplus G_n$ and $G_n \triangleleft G_{n-1} \triangleleft \ldots \triangleleft G_1$ are unambiguously defined.

\section{Synthesis of linear quantum stochastic systems}
\label{sec:synthesis-problem}
The network synthesis problem for linear quantum stochastic
systems can be stated (in a strict sense, as explained below) 
as the problem of how to systematically
realize a given linear quantum stochastic system with a given
fixed set of matrix parameters $S,K,R$ from a bin of certain basic
quantum optical components; see \cite{NJD08} for details of these
basic components. A particular solution was proposed to the
synthesis problem, see \cite[Theorem 5.1]{NJD08}: Any linear
quantum stochastic system with $n$ degrees of freedom can be
synthesized via a quantum network consisting of a cascade
connection of $n$ one degree of freedom harmonic oscillators
together with a direct interaction Hamiltonian that is bilinear in
the canonical operators of the oscillators. Partition $R$ as
$R=[R_{jk}]_{j,k=1,\ldots,n}$ with $R_{jk} \in \Rbb^{2 \times
2}$ and $K$ as $K=[\begin{array}{cccc} K_1 & K_2 &\ldots &K_n
\end{array}]$ with $K_k \in \Cbb^{m \times 2}$. Then according to
\cite[Theorem 5.1]{NJD08} a system $G=(S,Kx,\half x^T R x)$ can be
decomposed as $G = (G_n \triangleleft G_{n-1} \triangleleft \cdots \triangleleft
G_1) \boxplus (0,0,H^d)$, where the $G_i$'s are (simpler) one degree of freedom open
harmonic oscillators $G_i=(S_i,K_ix_i,\half x_i^T R_{ii} x_i)$
($x_i=(q_i,p_i)^T$) with parameter values specified by the theorem,
and $H^d$ is a direct bilinear interaction Hamiltonian of the
form $H^d=\sum_{j=1}^{n-1}\sum_{k=j+1}^{n} x_j^T\left( R_{jk}-
\Im\{K_{k}^{\dag}K_j\}^T \right)x_k$. The work \cite{NJD08} then shows how each of the $G_i$'s can be
synthesized from the bin of given components and how $H^d$ can be
realized. However, in current practical experiments,
implementation of $H^d$ can be challenging for systems that have
more than just a few degrees of freedom. Therefore, it is of
interest to characterize the class of systems that can be
synthesized by pure cascade connection alone, that is, with $H^d \equiv 0$.

As alluded to at the beginning of this section, we emphasize that \cite{NJD08} considers a {\em strict}
type of realization problem, that is, it deals with how to
synthesize a {\em given} and {\em fixed} triplet
$\{S,L=Kx,H=\frac{1}{2}x^T R x\}$ that describes a linear quantum
stochastic system $G$. This type of strict realizability is
relevant, for instance, in cases where the internal dynamics
$X(t)$ may represent some (continuous time) quantum information
processing algorithm and thus needs to be realized as given.
However, for some linear quantum control problems such as robust
disturbance attenuation \cite{JNP06} and LQG synthesis
\cite{NJP07b}, the internal dynamics are
inconsequential. In this case there is freedom to
modify/transform these dynamics and what is important is the
associated (classical) complex transfer function associated with
the system matrices $(A,B,C,D)$\footnote{As in \cite{JNP06}, here
we shall not define the transfer function of quantum systems, but associate to
 a quantum system $G$ with system matrices $(A,B,C,D)$ a {\em classical}, {\em doubled-up} \cite{GJN10}, transfer function
$G(s)=[\begin{array}{cc} C^T & C^{\dag}\end{array}]^T(sI-A)^{-1}B+\diag(D,D^{\#})$. However, we also remark that $G(s)$ can actually be properly interpreted as a genuine transfer function for the quantum system following \cite{GJN10,GGY08,YK03a,YK03b}, this being a common practice in the physics community via Fourier transform methods \cite{WaM94,GZ04}. In any case, we are dealing with the same object $G(s)$ and thus the particular interpretation attached to it becomes immaterial for our purpose.}. As is well known, a transfer function is
invariant under a similarity transformation of the system matrices
$(A,B,C,D) \mapsto (VAV^{-1},VB,CV^{-1},D)$ for any invertible
matrix $V$. However, for linear quantum systems the transformation
matrix $V$ for a similarity transformation is restricted in that
it has to be a {\em symplectic} matrix: $V$ is real and satisfies the
condition $V\Theta V^T=\Theta$. This ensures that the
transformed variable $Z(t)=VX(t)$ satisfies the required canonical commutation 
relations (CCR) of quantum mechanics: $Z(t)Z(t)^T-(Z(t)Z(t)^T)^T=2\isym \Theta$, so the system remains
physical. Note that the set of all symplectic matrices of a fixed
dimension form a group and in particular $V^{-1}$ is again a
symplectic matrix. Such a similarity transformation in quantum
systems corresponds to replacing $G=(S,Kx,\half x^T Rx)$ with
$G'=(S,KV^{-1}x,\half x^T V^{-T}RV^{-1}x)$. This motivates us to
introduce the following definition:

\begin{definition}
\label{df:io-equiv} Let $G=(S,Kx,\half x^T R x)$ and
$G'=(S',K'x,\half x^T R' x)$ be two linear quantum stochastic
systems. Then $G'$ is said to be {\em transfer function
equivalent} to $G$ or is a {\em transfer function realization} of
$G$ if  $S'=S$ and there exists a symplectic matrix $V$
such that $R'=V^{-T}RV^{-1}$, $K'=KV^{-1}$ (or, equivalently, $R=V^{T}R'V$
and $K=K'V$). $G$ is then said to be {\em transfer function
realizable} by $G'$, and vice-versa.
\end{definition}

\begin{remark}
It is important to note that two transfer function equivalent
systems $G$ and $G'$ will not necessarily generate the same
input-output dynamics $(A(t),Y(t))$ for all $t \geq 0$. This is
because although they can have different parameters, they always
have the same initial value $X(0)=x$, whilst for quantum systems $x$
clearly cannot be zero due to the CCR condition. If the $A$ matrix of $G$ is Hurwitz then
the input-output dynamics of $G$ and $G'$ converge in the limit $t
\rightarrow \infty$. However, as remarked earlier, for some linear
quantum control design objectives internal dynamics and initial
conditions do not play an essential role, only the transfer
function does.
\end{remark}

\section{Conditions for realizability by a pure cascade connection}
\label{sec:cascade-conds}
In this section we state and prove a theorem that characterizes
the class of linear quantum stochastic systems that can be
realized simply by a cascade connection of one degree of freedom (open quantum harmonic) oscillators. Let us first
introduce the following notation: $S_{k \twoheadleftarrow
j}=S_k \cdots S_{j+1} S_j$ for all $j<k$, $S_{k \twoheadleftarrow
k}=S_k$ and $S_{k \twoheadleftarrow k+1}=I_{m}$, and let
$x_i=(q_i,p_i)^T$ for $i=1,\ldots,n$ so that
$x=(x_1^T,\ldots,x_n^T)^T$, where $xx^T-(xx^T)^T=2\isym \Theta$.
Moreover, we introduce the following terminology: A square matrix
$F$ is said to be lower $2 \times 2$ block triangular if it has a
lower block triangular form when partitioned into $2 \times 2$
blocks:
\begin{eqnarray*}
F= \left[\begin{array}{ccccc} F_{11} & 0_{2 \times 2} & 0_{2
\times 2} & \ldots & 0_{2 \times 2}\\
F_{21} & F_{22} & 0_{2 \times 2} & \ldots & 0_{2 \times 2}\\
\vdots & \ddots & \ddots & \ddots & \vdots \\
F_{n1} & F_{n2} & \ldots & \ldots& F_{nn}
\end{array}\right],
\end{eqnarray*}
where $F_{jk}$, $j\leq k$, is of dimension $2 \times 2$. We start
with the following lemma:

\begin{lemma}
\label{lm:casc-struc} The cascade connection $G_n \triangleleft
G_{n-1} \triangleleft \cdots \triangleleft G_1$ of one degree of
freedom harmonic oscillators $G_i=(S_i,K_ix_i,\frac{1}{2}x_i^T R_i
x_i)$ $(i=1,\ldots,n)$ realizes a linear quantum stochastic system
$G=(S,Kx,\frac{1}{2}x^T R x)$ with
$
S=S_{n \twoheadleftarrow 1},\;
K = \left[\begin{array}{ccccc} S_{n \twoheadleftarrow 2} K_1 &
S_{n\twoheadleftarrow 3} K_2 & \ldots & K_n\end{array}\right],\;
R = [R_{ij}]_{i,j=1,\ldots,n}$,
where $R_{jj}=R_j$, $R_{kj}=\Im\{K_k^{\dag}S_{ k \twoheadleftarrow
 j+1}K_j\}$ whenever $k<j$ and $R_{jk}=R_{kj}^T$ whenever $j>k$.
In particular, $R+\Im\{K^{\dag}K\}$ is lower $2 \times 2$ block
triangular.
\end{lemma}
\begin{proof}
The proof proceeds along the lines of the proof of \cite[Theorem
5.1]{NJD08}. By the series product formula (cf. section \ref{sec:LQSS-prelim}) for the cascade of two one degree of freedom
oscillators $G_1=(S_1,K_1x_1,\frac{1}{2}x_1^T R_1
x_1)$ and $G_2=(S_2,K_2 x_2,\frac{1}{2}x_2^T R_2 x_2)$, we get the oscillator $G_{(2)}=G_2 \triangleleft G_1 = (S_2S_1, S_2K_1x_1+K_2 x_2,
\frac{1}{2}x_1^T R_1x_1 + \frac{1}{2} x_2^T R_2x_2 +
x_2^T\Im\{K_2^{\dag}S_2K_1\}x_1)$. Letting
$x_{(2)}=(x_1^T,x_2^T)^T$, the latter may be compactly written as:
$G_{(2)}=(S_{(2)},K_{(2)}x_{(2)},\frac{1}{2}x_{(2)}^T R_{(2)}
x_{(2)}^T)$ with $S_{(2)}=S_{2 \twoheadleftarrow 1}=S_2 S_1$,
$K_{(2)}= [\begin{array}{cc} S_2K_1 & K_2 \end{array}]$ and
$
R_{(2)}=\left[\small \begin{array}{cc} R_1 & \Im\{K_2^{\dag}S_{2 \twoheadleftarrow 2}K_1\}^T  \\
\Im\{K_2^{\dag}S_{2 \twoheadleftarrow 2}K_1\} & R_{2}
\end{array} \normalsize \right]
$. Repeating the computation for $G_{(k)}=G_{k} \triangleleft G_{(k-1)}$
iteratively for $k=3,\ldots,n-1$ and writing
$x_{(k)}=(x_1^T,x_2^T,\ldots,x_k^T)^T$ and
$G_{(k)}=(S_{(k)},K_{(k)}x_{(k)},\frac{1}{2}x_{(k)}^TR_{(k)}x_{(k)})$
at each iteration $k$, we arrive at the desired result with $G=G_{(n)}$,
$S=S_{(n)}$, $K=K_{(n)}$ and $R=R_{(n)}$ as stated in the lemma.

To see that $R+\Im\{K^{\dag}K\}$ is lower $2 \times 2$ block triangular, we note that $K^{\dag}K$ may
be expressed as follows:
\begin{eqnarray*}
K^{\dag}K 
&=&\left[\small \begin{array}{ccccc} K_1^{\dag}K_1 &
K_1^{\dag}S_2^{\dag}K_2 & K_1^{\dag}S_{3 \twoheadleftarrow 2}^{\dag}K_2 & \ldots & K_1^{\dag}S_{n \twoheadleftarrow 2}^{\dag}K_n\\
K_2^{\dag}S_2 K_1 & K_2^{\dag}K_2 & K_2^{\dag}S_{3}^{\dag}K_3 & \ldots &K_2^{\dag}S_{n \twoheadleftarrow 3}^{\dag}K_n\\
\vdots & \ddots & \ddots & \ddots & \vdots\\
K_n^{\dag}S_{n \twoheadleftarrow 2}K_1 & K_n^{\dag}S_{n
\twoheadleftarrow 3}K_2 & \ldots & K_n^{\dag}S_n K_{n-1} & K_n^{\dag}K_n \end{array}\normalsize \right].\\
\end{eqnarray*}
Note that since $K^{\dag}K$ is by definition a Hermitian matrix,
the $2 \times 2$ block elements above the diagonal blocks are the
Hermitian transpose of the corresponding elements below the
diagonal blocks. It follows therefore that the imaginary part of
the block $(K^{\dag}K)_{jk}$ at block row $j$ and block column $k$
must satisfy the relation:
$\Im\{(K^{\dag}K)_{jk}\}=-\Im\{(K^{\dag}K)_{kj}\}^T$. However, from
the expression for $R$ derived above and its symmetry, we already
have that if $k>j$:
$$R_{jk}=R_{kj}^T=\Im\{K_k^{\dag}S_{k \twoheadleftarrow j+1} K_j\}^T=\Im\{(K^{\dag}K)_{kj}\}^T.$$
Therefore, the off-diagonal upper block elements of $R$ cancel 
those of $\Im\{K^{\dag}K\}$ when they are summed and we conclude
that the matrix $R+\Im\{K^{\dag}K\}$ is a lower $2 \times 2$ block
triangular matrix.
\end{proof}

Recall again the partitioning of $R$ as $R=[R_{jk}]_{j,k=1,\ldots,n}$
with $R_{jk} \in \Rbb^{2 \times 2}$ and of $K$ as
$K=[\begin{array}{cccc} K_1 & K_2 &\ldots &K_n
\end{array}]$ with $K_k \in \Cbb^{m \times 2}$. We may now state the
following result:

\begin{theorem}
\label{th:casc-struc} A linear quantum stochastic system
$G=(S,Kx,\frac{1}{2}x^T R x)$ with $n$ degrees of freedom is
realizable by a pure cascade of $n$ one degree of freedom harmonic
oscillators (without a direct interaction Hamiltonian) if and only if the $A$
matrix given by $A=2\Theta(R+\Im\{K^{\dag}K\})$ is a lower block
triangular matrix with blocks of size $2 \times 2$. If this
condition is satisfied then $G$ can be explicitly constructed as
the cascade connection $G_n \triangleleft G_{n-1} \triangleleft
\ldots \triangleleft G_1$ with $G_1=(S,K_1x_1,\half x_1^T
R_{11}x_1)$, and $G_k=(I,K_kx_k,\half x_k^T R_{kk} x_k)$ for
$k=2,\ldots,n$.
\end{theorem}
\begin{proof}
The proof of the only if part follows directly from Lemma
\ref{lm:casc-struc}, as follows. If $G$ can be realized by a pure
cascade connection of $n$ one degree of freedom harmonic
oscillators then by the lemma, $R+\Im\{K^{\dag}K\}$ is a lower
$2 \times 2$ block triangular matrix. However, since $\Theta$ is $2 \times 2$ block
diagonal, it follows that the matrix
$A=2\Theta(R+\Im\{K^{\dag}K\})$ is also a lower $2 \times 2$ block triangular
matrix.

Conversely, the if part of the proof can be shown by explicitly
constructing a pure cascade connection of $n$ one degree
oscillators that realizes $G$. If the $A$ matrix associated with
$G$ is lower $2 \times 2$ block triangular then so is the matrix $\half
\Theta^{-1}A=-\half \Theta A=R+\Im\{K^{\dag}K\}$. As we already saw
in the proof of Lemma \ref{lm:casc-struc}, this structure implies that $R_{jk}=\Im\{K_j^{\dag}K_k\}$
whenever $k>j$ and $R_{kj}=\Im\{K_j^{\dag}K_k\}^T$ if $k<j$. Now,
using the notation of Lemma \ref{lm:casc-struc}, let us define the
one degree of freedom harmonic oscillators $G_k$ for
$k=1,\ldots,n$ as $G_1=(S,K_1 x_1,\half  x_1^T R_{11} x_1)$, and
$G_k=(I,K_k x_k,\half x_k^T R_{kk}x_k)$ for $k=2,\ldots,n$. It follows
from  Lemma \ref{lm:casc-struc} that $G_n \triangleleft
G_{n-1} \triangleleft \cdots \triangleleft G_1 = (S,K x, \half x^T R
x)$. That is, this cascade connection realizes $G$.
\end{proof}

Theorem \ref{th:casc-struc} has a direct consequence on the weaker
notion of transfer function realization of a linear quantum
system. The main result is the following corollary:

\begin{corollary}
\label{cl:casc-struc}A linear quantum system $G=(S,Kx,\half x^T R x)$ is
transfer function realizable by a pure cascade connection of one
degree of freedom harmonic oscillators if and only if there is a
symplectic transformation matrix $V$ such that the linear quantum
stochastic system $G'=(S,KV^{-1} x, \half x^T V^{-T}RV^{-1}x)$ has an
$A$ matrix which is lower $2 \times 2$ block triangular.
\end{corollary}
\begin{proof}
By Definition \ref{df:io-equiv}, $G$ is transfer function realizable by a pure cascade
connection if and only if there exists a symplectic matrix $V$ such
that $G'=(S,KV^{-1} x, \half x^T
V^{-T}RV^{-1}x)$ is realizable by a pure cascade
connection. But from Theorem
\ref{th:casc-struc} this is true if and only if the $A$ matrix
associated with $G'$ (i.e., $A=2\Theta(R'+\Im\{K'^{\dag}K'\})$ is
lower $2 \times 2$ block triangular. \end{proof}

\section{Passive linear quantum stochastic systems}
\label{sec:passive-sys}
In this section it will be shown that the class of passive linear
quantum stochastic systems (as defined below) are transfer function realizable by a
cascade connection. In \cite{Pet09} it has been  shown by a constructive algorithm that a ``generic'' sub-class of such systems are transfer function realizable by pure cascading, the generic systems being required to satisfy assumptions on the distinctness of the eigenvalues and invertibility of certain matrices. In this section we remove such assumptions, and show by exploiting the algebraic structure of passive systems that the result is valid in general for {\em all} passive  systems.

For $k=1,\ldots,n$, let $a_k=(q_k+\isym p_k)/2$ be the
annihilation operators for mode $k$ and define
$a=(a_1,\ldots,a_n)^T$. Then $a$ satisfies the CCR
$
\left[\begin{array}{c} a \\
a^{\#}\end{array}\right] [
\begin{array}{cc} a^{\dag} & a^T \end{array} ]-\left(\left[\begin{array}{c} a^{\#} \\ a \end{array}\right] [
\begin{array}{cc} a^T & a^{\dag} \end{array} ]\right)^T = \diag(I_n,-I_n).
$
Moreover, note that $(a^T,a^{\dag})^T=[\begin{array}{cc} \Sigma^T & \Sigma^{\dag}\end{array}]^Tx$ with 
$$
\Sigma =\left[\begin{array}{ccccccc} \half & \half \isym & 0 & 0 & 0 &
\ldots & 0\\
0 & 0 & \half & \half \isym & 0 &  \ldots & 0\\
\vdots & \ddots & \ddots & \ddots & \ddots & \ddots & \vdots \\
0 & \ldots & \ldots & \ldots & 0 & \half & \half \isym
\end{array}\right].
$$
We also make note that $\left[\begin{array}{c} \Sigma \\
\Sigma^{\#} \end{array} \right]^{-1}=2\, [\begin{array}{cc}
\Sigma^{\dag} & \Sigma^{T} \end{array}]$ and from the
relation $\left[\begin{array}{c} \Sigma \\
\Sigma^{\#} \end{array} \right]2\,[\begin{array}{cc}
\Sigma^{\dag} & \Sigma^{T} \end{array} ]=I$ we have the
identities:
\begin{eqnarray}
\Sigma \Sigma^{\dag}=I/2=\Sigma^{\#}\Sigma^T ;\;
\Sigma \Sigma^{T}=0=\Sigma^{\#}\Sigma^{\dag}. \label{eq:Sigma-id}
\end{eqnarray}
Therefore, we also have
$$
x= \left[\begin{array}{c} \Sigma \\
\Sigma^{\#} \end{array} \right]^{-1}\left[\begin{array}{c} a \\
a^{\#}\end{array}\right]=2\, [\begin{array}{cc} \Sigma^{\dag} &
\Sigma^{T} \end{array} ]\,\left[\begin{array}{c} a \\
a^{\#}\end{array}\right].
$$
A system $G=(S,Kx,\half x^T R x)$ is said to be {\em passive} if
we can write $H=\half x^T Rx = \half a^{\dag} \tilde R a+c$ and
$L=Kx=\tilde K a$ for some complex $n \times n$ {\em Hermitian}
matrix $\tilde R$ , a complex $m \times n$ (here $m$ again denotes the
number of input and output fields in and out of $G$) matrix
$\tilde K$, and some {\em real} constant $c$. As discussed in \cite{Nurd09b}, here the term passive for such systems is physically motivated since they can be implemented using only passive components like optical cavities, mirrors, beamsplitters and phase shifters; this  follows from Theorem 5.1 of \cite{NJD08} and the constructions shown in section 6 of that paper. Also shown in \cite{Nurd09b}, we can express $\half a^{\dag} \tilde R a$ and $\tilde K a$ in the form $
\half a^{\dag} \tilde R a =  \half  x^T \Re\{\Sigma^{\dag}\tilde R\Sigma\}x - \frac{1}{4} \sum_{j=1}^n \tilde R_{jj}$ and $\tilde K a=  \tilde K\Sigma x$. Therefore, we may set $R = \Re\{\Sigma^{\dag}\tilde R\Sigma\}$ and $K = \tilde K \Sigma$. Note also from \cite{Nurd09b} that the $2 \times 2$ block diagonal elements $\{R_{jj};\;j=1,\ldots,n\}$ is of the form $R_{jj}=\lambda_jI_{2}$ for some $\lambda_j \in \Rbb$ for all $j$. Now we shall derive some properties of $A$ and show that there exists a unitary and symplectic matrix that transforms it into a
 lower  $2 \times 2$ block triangular matrix.
\begin{lemma}
\label{lm:A-decomp}$ [\begin{array}{cc} \Sigma^T  & \Sigma^{\dag}\end{array}]^T
 A\, [\begin{array}{cc} \Sigma^{\dag} &
\Sigma^T
\end{array} ]=\diag(M,M^{\#})$, where
$
M=\half \Sigma \Theta \Sigma^{\dag}(\tilde R-\isym \tilde K^{\dag}
\tilde K)$.
\end{lemma}
\begin{proof}
For the proof, we exploit the identities (\ref{eq:Sigma-id}) as
well as the following easily verified identities:
$
\Sigma \Theta \Sigma^{T}=0= \Sigma^{\#}\Theta \Sigma^{\dag}$.
Using these identities, we have the following:
\begin{eqnarray*}
[\begin{array}{cc} \Sigma^T & \Sigma^{\dag}
\end{array}]^T A\,[\begin{array}{cc} \Sigma^{\dag} &
\Sigma^T
\end{array}] &=& 2[\begin{array}{cc} \Sigma^T & \Sigma^{\dag}
\end{array}]^T \Theta(R + \Im\{K^{\dag}K\})\,[\begin{array}{cc}
\Sigma^{\dag} &
\Sigma^T \end{array}]\\
&=& 2[\begin{array}{cc} \Sigma^T & \Sigma^{\dag}
\end{array}]^T \Theta(\Re\{\Sigma^{\dag}\tilde R\Sigma\} +
\Im\{\Sigma^{\dag}\tilde K^{\dag} \tilde K\Sigma\}) \,[\begin{array}{cc}
\Sigma^{\dag} & \Sigma^T \end{array}]\\
&=& 2[\begin{array}{cc} \Sigma^T & \Sigma^{\dag}
\end{array}]^T \Theta\biggl(\half(\Sigma^{\dag}\tilde R\Sigma+ \Sigma^{T}\tilde
R^{\#}\Sigma^{\#}) \\
&&\quad -\frac{\isym}{2}(\Sigma^{\dag}\tilde K^{\dag}\tilde K\Sigma -
\Sigma^{T}\tilde K^{T}\tilde
K^{\#}\Sigma^{\#})\biggr)\,[\begin{array}{cc}
\Sigma^{\dag} & \Sigma^T \end{array}]\\
&=& 2[\begin{array}{cc} \Sigma^T & \Sigma^{\dag}
\end{array}]^T \Theta \left[\begin{array}{cc} \quarter \Sigma^{\dag}\tilde
R-\frac{\isym}{4} \Sigma^{\dag} \tilde K^{\dag}\tilde K & \quarter \Sigma^T \tilde R^{\#}
+\frac{\isym}{4}\Sigma^T \tilde K^{T} \tilde K^{\#} \end{array}\right]\\
&=&\diag\biggl(\half \Sigma \Theta\Sigma^{\dag}\tilde
R-\frac{\isym}{2}\Sigma \Theta\Sigma^{\dag}\tilde K^{\dag}\tilde K, \\
&&\quad \half
\Sigma^{\#}\Theta\Sigma^T \tilde R^{\#} +\frac{\isym}{2}\Sigma^{\#}\Theta
\Sigma^T \tilde K^{T} \tilde  K^{\#}\biggr).
\end{eqnarray*}
\end{proof}

Then we have the following theorem:
\begin{theorem}
\label{th:passive-casc}Let $U$ be the complex unitary matrix in
a {\em Schur decomposition} of the matrix $M$ of Lemma \ref{lm:A-decomp}: $M=U^{\dag}\hat M U$, where
$\hat M$ is a lower triangular matrix. Then the matrix
$$
V= 2\,[\begin{array}{cc} \Sigma^{\dag} & \Sigma^T \end{array}]\,\diag(U,U^{\#})[\begin{array}{cc} \Sigma^T & \Sigma^{\dag}\end{array}
]^T
$$
is a {\em real}, {\em unitary}, and {\em symplectic} matrix that
transforms $A$ into a lower $2 \times 2$ block triangular matrix:
$VAV^{\dag}=\hat A$, where $\hat A$ is  a real lower $2 \times 2$
block triangular matrix. Therefore, every passive linear quantum
system has a transfer function realization by pure cascading and
such a realization is obtained by applying the construction of
Theorem \ref{th:casc-struc} to $G'=(S, KV^Tx , \half x^T VRV^T x)$. Moreover, each
of the one degree of freedom oscillator in the cascade will also
be passive.
\end{theorem}
\begin{proof}
The existence of $U$ is guaranteed by the well known result that
every complex matrix $M$ has a Schur decomposition of the form
$M=U^{\dag}\hat M U$ with $\hat M$ lower triangular. Note then that we
also have $\hat M^{\#}=U^{\#} M^{\#}U^T$. Let $V$ be as defined
in the theorem. Then by Lemma \ref{lm:A-decomp} the following is
true:
\begin{eqnarray*}
VAV^{\dag} = 4\, [\begin{array}{cc} \Sigma^{\dag} & \Sigma^T
\end{array} ]\,\diag(\hat M,\hat M^{\#})[\begin{array}{cc} \Sigma^T & \Sigma^{\dag}\end{array}
]^T = 4(\Sigma^{\dag}\hat M\Sigma + \Sigma^T \hat M^{\#}\Sigma^{\#})= 8\Re\{  \Sigma^{\dag}\hat M\Sigma \}.
\end{eqnarray*}
Now, since $\hat M$ is a lower triangular matrix, it follows by
inspection (using the special structure of $\Sigma$) that $\hat A=8\Re\{ \Sigma^{\dag}\hat M\Sigma\}$ is a lower $2
\times 2$ block diagonal matrix, as claimed. That $V$ is real follows from the fact that we may write
$V=2(\Sigma^{\dag}U\Sigma
+\Sigma^TU^{\#}\Sigma^{\#})=4\Re\{\Sigma^{\dag}U\Sigma \}$. That it
is {\em unitary} follows from the observation that
$\sqrt{2}[\begin{array}{cc} \Sigma^{\dag} & \Sigma^T
\end{array}]$ and $\sqrt{2}\left[\begin{array}{c}
\Sigma \\ \Sigma^{\#}
\end{array}\right]$ are unitary (as a consequence of (\ref{eq:Sigma-id})) and that $\diag(U,U^{\#})$ is also unitary. To see that $V$ is also {\em symplectic}
define $b=Ua$ and $z=Vx$. By the unitarity of $U$ we have that $b$
and $b^{\#}$ satisfy the same the commutation relations as $a$ and
$a^{\#}$ (i.e., $b$ is again an annihilation operator). Then we
have
\begin{eqnarray*}
z=Vx =V2[\begin{array}{cc} \Sigma^{\dag} & \Sigma^T
\end{array} ]\left[ \begin{array}{c} a \\ a^{\#}
\end{array}\right]= 2[\begin{array}{cc} \Sigma^{\dag} & \Sigma^T \end{array}
] \diag(U,U^{\#}) \left[ \begin{array}{c} a \\ a^{\#}
\end{array}\right]=[\begin{array}{cc} \Sigma^{\dag} & \Sigma^T \end{array}] \left[ \begin{array}{c} b \\ b^{\#}
\end{array}\right].
\end{eqnarray*}
Now, this implies that $z$ consists of the canonical position and
momentum operators associated with the modes in $b$ and satisfies
the same CCR as $x$. But since $z=Vx$
and $V$ is real, preservation of the CCR implies that $V$ is
necessarily a symplectic matrix (this is standard knowledge in quantum mechanics; see, e.g., \cite[section III]{lin98}).

Using the fact that $V^{-1}=V^{T}=V^{\dag}$ established above, it
follows from Theorem \ref{th:casc-struc}  that the passive quantum
system $G$ is transfer function equivalent to 
$G'=(S,KV^Tx,\half x^T V RV^{T} x)$
whose $A$ matrix is lower $2
\times 2$ block triangular. Let $K'=KV^T=[\begin{array}{cccc}
K'_1 & K'_2 & \ldots & K'_n
\end{array}]$ and $R'=VRV^T=[R'_{jk}]_{j,k=1,\ldots,n}$. By
Theorem \ref{th:casc-struc} we have $G'=G_n \triangleleft G_{n-1}
\triangleleft \ldots \triangleleft G_1$ with
$G_k=(S_k,K'_kx_k,\half x_k^T R_{kk}'x_k)$, $S_1=S$ and $S_k=I$
for $k>1$. We now show that each $G_k$ is passive. Recall that $K=\tilde K\Sigma$ and write
$\tilde K=[\begin{array}{ccc} \tilde K_1 & \ldots & \tilde K_n
\end{array}]$ with $\tilde K_k \in \Cbb^{m \times 1}$. Using
(\ref{eq:Sigma-id}) we have that $K'x=KV^{\dag}x=\tilde KU^{\dag} a$. By expanding
both sides of the equality $K'x=\tilde KU^{\dag}a$ and collecting and
equating terms of the same index, it follows that $K'_kx_k=
(\tilde KU^{\dag})_k a_k$, where $(\tilde KU^{\dag})_k$ is the $k$-th $\Cbb^{m \times 1}$ block component of $\tilde KU^{\dag}$. On the other hand, since $G$ is passive we have that $R_{kk}=\lambda_k I_2$
for some $\lambda_k \in \Rbb$, and recalling that $R=\Re\{\Sigma^{\dag} \tilde R \Sigma\}$, it follows by inspection (after some algebraic manipulations using (\ref{eq:Sigma-id})) that $R'_{kk}=\lambda'_{k}I_2$  for some $\lambda'_k  \in \Rbb$. Thus, $x_k^T R'_{kk} x_k= \lambda_k' a_k^{*}a_k + c$ for some constant
$c$ and we conclude that each $G_k$ is also passive. 
\end{proof}

\begin{example}
Let $G=(I,\tilde K a,\half
a^{\dag}\tilde R a+\frac{5}{4})$ be a passive system with $\tilde R=\left[\small \begin{array}{cc} 2 & 1+\isym \\
1-\isym & 3\end{array}\normalsize \right] $ and $\tilde K=\left[\small \begin{array}{cc}
1+0.5\isym & -2+\isym\\ -5-2\isym & 3-4\isym \end{array} \normalsize \right]$.  Here $K=\left[\small \begin{array}{cccc} 0.5+0.25\isym 
& -0.25+0.5\isym & -1 + 0.5\isym & -0.5 -\isym\\ -2.5-\isym & 1-2.5\isym & 1.5-2\isym & 2+1.5\isym \end{array} \normalsize \right]$ and \\
$R=\left[\small \begin{array}{cc} 0.5 I_2 & 0.25 (I_2-J) \\ 0.25 (I_2+J) & 0.75I_2\end{array} \normalsize \right] $.  
\hspace{1cm}  By Lemma \ref{lm:A-decomp} and Theorem \ref{th:passive-casc}, we have that 
$U=\left[\small \begin{array}{cc} -0.6933 + 0.0039\isym &  0.2244 - 0.6849\isym
\\ 0.7204 + 0.0209\isym & 0.2312 - 0.6536\isym
\end{array} \normalsize \right]$ and $\hat M=\left[\small \begin{array}{cc}  -14.8390 - 0.7912\isym  & 0 \\ 0.6344 - 0.2225\isym & -0.2235 - 0.4588\isym
\end{array}\normalsize \right]$. Then by the formula of Theorem
\ref{th:passive-casc}  we have
$$V=\left[\small \begin{array}{cccc}
-0.6933 &  0.0039 & 0.2244 &  0.6849\\
   -0.0039 & -0.6933 & -0.6849 & 0.2244\\
   0.7204 & -0.0209 & 0.2312 & 0.6536\\
    0.0209 & 0.7204 & -0.6536 & 0.2312
\end{array}\normalsize \right].$$
Let $K'=KV^{-1}=[\begin{array}{cc} K'_1 & K'_2 \end{array}]$, then
\begin{align*}
K'&=\left[\small \begin{array}{cccc} -0.9144 - 0.7441\isym & 0.7441  - 0.9144\isym
& -0.1926 - 0.3684 \isym & 0.3684 - 0.1926\isym \\
3.4433 + 1.2621\isym  & -1.2621+ 3.4433\isym & -0.1679 - 0.1501\isym  & 0.1501 - 0.1679\isym \end{array} \normalsize \right],
\end{align*}
and
\begin{align*}
R' &=VRV^{-1}= \left[\begin{array}{cc} R'_{11} & R'_{12}\\ (R'_{12})^T &
R'_{22}
\end{array}\right] =\left[\small \begin{array}{cccc}
      0.7912  &   0  &  0.1113 &  0.3172\\
         0  &   0.7912 &  -0.3172  &   0.1113\\
    0.1113 &  -0.3172  &  0.4588 &   0\\
    0.3172 &   0.1113 &  0 &   0.4588
\end{array} \normalsize \right].
\end{align*}
Therefore, by the theorem, $G$ is transfer function realizable by
$G'=(I,K'x,\half x^T R' x)$. It is easily checked that
$R'+\Im\{K'K\}$ is lower $2 \times 2$ block triangular:
\begin{align*}
R'+\Im\{K'^{\dag}K'\}=\left[\small \begin{array}{cccc} 0.7912 & 14.8390 &
0 & 0\\
   -14.8390  &   0.7912 &   0   &  0\\
   0.2225 &  -0.6344 &   0.4588 &  0.2235\\
    0.6344 &  -0.2225 & -0.2235 &   0.4588
\end{array} \normalsize \right],
\end{align*}
therefore $G'$ can be realized by a pure cascade connection of one
degree of freedom harmonic oscillators. According to Theorem
\ref{th:casc-struc}, $G'=G_2 \triangleleft G_1$ with
$G_1=(I,K'_1x_1,\half x_1^T R'_{11}x_1)$ and $G_2=(I,K'_2x_2,\half
x_2^T R'_{22}x_2)$. It is easily inspected that both $G_1$
and $G_2$ are passive. A quantum optical realization of $G'$ is illustrated in Fig.~\ref{fig:passive-realization}.

\begin{figure}[h!]
\centering
\includegraphics[scale=0.4]{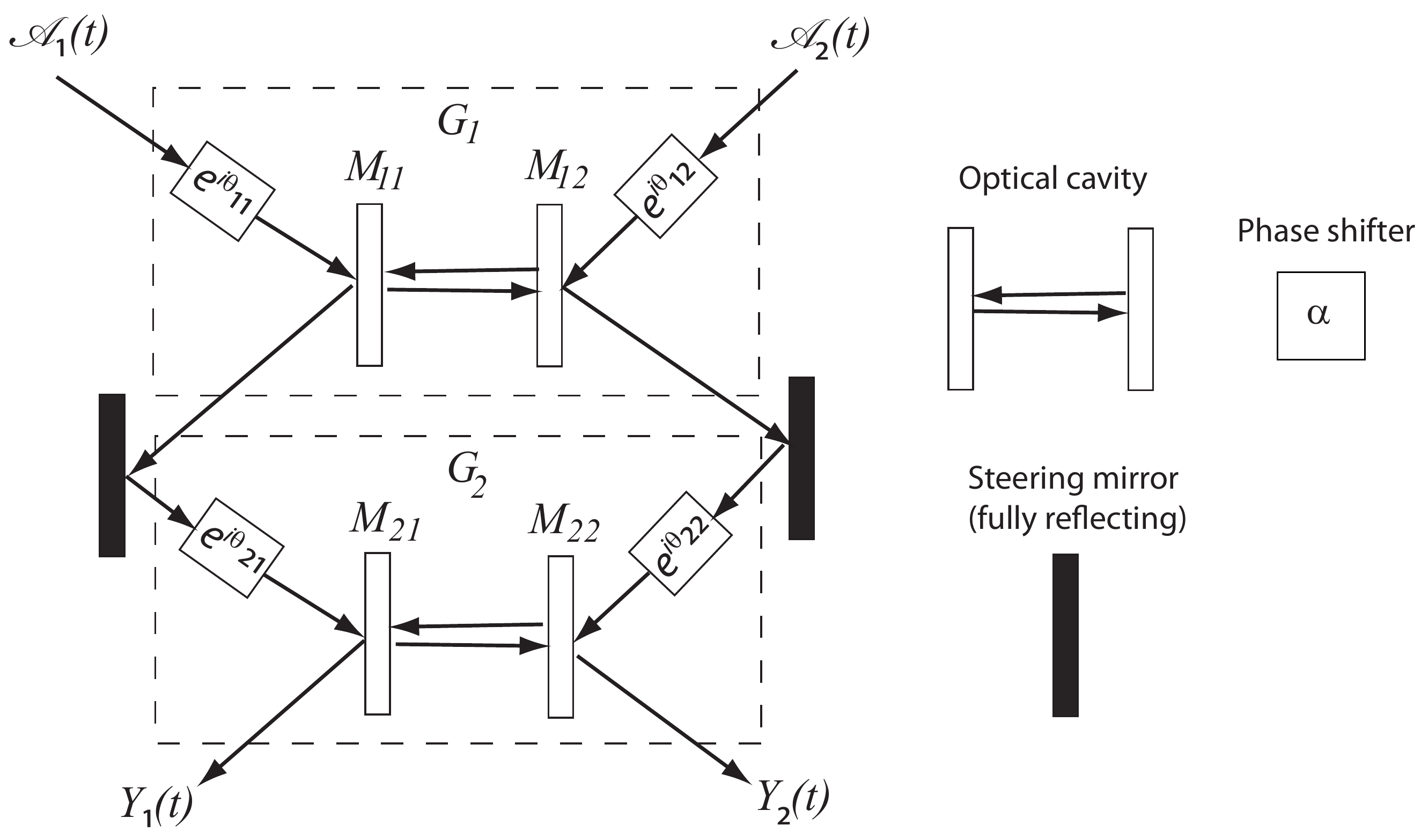}
\caption{Realization of $G'$ as the cascade connection of $G_1$ and $G_2$. $G_1$ and $G_2$ are each realized
by an optical cavity and a phase shifter; see, e.g., \cite{BR04,WaM94,NJD08} for a discussion of these devices. Dark rectangles depict  fully reflecting mirrors, while light rectangles depict partially transmitting mirrors; an optical cavity is formed by bouncing light back and forth between two mirrors. Here $\theta_{11}=-2.4585$, $\theta_{12}=0.3513$, $\theta_{21}=-2.0525$, and $\theta_{22}=-2.4121$, and the partially transmitting mirror $M_{jk}$ in the optical cavity  $G_j$, $j,k=1,2$, have the coupling coefficients  $\gamma_{11}=1.3898$, $\gamma_{12}=13.4492$,  $\gamma_{21}=0.1728$ and $\gamma_{22}=0.0507$, respectively. The resonance frequencies of the optical cavities of $G_1$ and $G_2$ have a detuning of $0.7912-0.4588=0.3324$, with $G_1$ having the higher resonance frequency.} \label{fig:passive-realization} 
\end{figure}
\end{example}

\section{Conclusions}
\label{sec:conclusions}
We have derived a characterization of linear quantum
stochastic systems that can be realized, in a strict or transfer
function sense, by a cascade connection of one degree of freedom
quantum oscillators alone, without requiring any direct bilinear interaction Hamiltonian between these oscillators. The
results are constructive in that if a system can be
realized by a cascade connection, it is explicitly shown how to
construct it. Then it was shown that the sub-class of passive
linear quantum stochastic systems is  always transfer function realizable by a pure
cascade connection.

\bibliographystyle{ieeetran}
\bibliography{ieeeabrv,rip,mjbib2004}

\begin{thebibliography}{10}
\providecommand{\url}[1]{#1}
\def\UrlFont{\rmfamily}
\providecommand{\newblock}{\relax}
\providecommand{\bibinfo}[2]{#2}
\providecommand\BIBentrySTDinterwordspacing{\spaceskip=0pt\relax}
\providecommand\BIBentryALTinterwordstretchfactor{4}
\providecommand\BIBentryALTinterwordspacing{\spaceskip=\fontdimen2\font plus
\BIBentryALTinterwordstretchfactor\fontdimen3\font minus
  \fontdimen4\font\relax}
\providecommand\BIBforeignlanguage[2]{{%
\expandafter\ifx\csname l@#1\endcsname\relax
\typeout{** WARNING: IEEEtran.bst: No hyphenation pattern has been}%
\typeout{** loaded for the language `#1'. Using the pattern for}%
\typeout{** the default language instead.}%
\else
\language=\csname l@#1\endcsname
\fi
#2}}

\bibitem{YK03b}
M.~Yanagisawa and H.~Kimura, ``Transfer function approach to quantum
  control-part {II}: Control concepts and applications,'' \emph{IEEE Trans.
  Automatic Control}, vol.~48, no.~12, pp. 2121--2132, 2003.

\bibitem{JNP06}
M.~R. James, H.~I. Nurdin, and I.~R. Petersen, ``{$H^{\infty}$} control of
  linear quantum stochastic systems,'' \emph{{IEEE} Trans. Automat. Contr.},
  vol.~53, no.~8, pp. 1787--1803, 2008.

\bibitem{NJP07b}
H.~I. Nurdin, M.~R. James, and I.~R. Petersen, ``Coherent quantum {LQG}
  control,'' \emph{Automatica J. IFAC}, vol.~45, pp. 1837--1846, 2009.

\bibitem{Mab08}
H.~Mabuchi, ``Coherent-feedback quantum control with a dynamic compensator,''
  \emph{Phys. Rev. A}, vol.~78, pp. 032\,323--1--032\,323--5, 2008.

\bibitem{NC00}
M.~Nielsen and I.~Chuang, \emph{Quantum Computation and Quantum
  Information}.\hskip 1em plus 0.5em minus 0.4em\relax Cambridge: Cambridge
  University Press, 2000.

\bibitem{Kimb08}
H.~J. Kimble, ``The quantum internet,'' \emph{Nature}, vol. 453, pp.
  1023--1030, 2008.

\bibitem{NJD08}
H.~I. Nurdin, M.~R. James, and A.~C. Doherty, ``Network synthesis of linear
  dynamical quantum stochastic systems,'' \emph{SIAM J. Control Optim.},
  vol.~48, no.~4, pp. 2686--2718, 2009.

\bibitem{AV73}
B.~D.~O. Anderson and S.~Vongpanitlerd, \emph{{Network Analysis and Synthesis:
  A Modern Systems Theory Approach}}, ser. Networks Series.\hskip 1em plus
  0.5em minus 0.4em\relax Prentice-Hall, Inc., 1973.

\bibitem{Nurd09b}
H.~I. Nurdin, ``Synthesis of linear quantum stochastic systems via quantum
  feedback networks,'' \emph{{IEEE} Trans. Automat. Contr.}, vol.~55, no.~4,
  pp. 1008--1013, 2010, extended preprint version available at
  http://arxiv.org/abs/0905.0802.

\bibitem{Pet09}
I.~R. Petersen, ``Cascade cavity realization for a class of complex transfer
  functions arising in coherent quantum feedback control,'' in
  \emph{Proceedings of the 2009 European Control Conference (Budapest, Hungary,
  23-26 August 2009)}, 2009.

\bibitem{HP84}
R.~L. Hudson and K.~R. Parthasarathy, ``{Quantum Ito's formula and stochastic
  evolution},'' \emph{Commun. Math. Phys.}, vol.~93, pp. 301--323, 1984.

\bibitem{GC85}
C.~W. Gardiner and M.~J. Collett, ``Input and output in damped quantum systems:
  {Q}uantum stochastic differential equations and the master equation,''
  \emph{Phys. Rev. A}, vol.~31, no.~6, pp. 3761 -- 3774, 1985.

\bibitem{KRP92}
K.~Parthasarathy, \emph{An Introduction to Quantum Stochastic Calculus}.\hskip
  1em plus 0.5em minus 0.4em\relax Berlin: Birkhauser, 1992.

\bibitem{GZ04}
C.~W. Gardiner and P.~Zoller, \emph{Quantum Noise: A Handbook of Markovian and
  Non-Markovian Quantum Stochastic Methods with Applications to Quantum
  Optics}, 3rd~ed.\hskip 1em plus 0.5em minus 0.4em\relax Berlin and New York:
  Springer-Verlag, 2004.

\bibitem{WM10}
H.~M. Wiseman and G.~J. Milburn, \emph{Quantum Measurement and Control}.\hskip
  1em plus 0.5em minus 0.4em\relax Cambridge University Press, 2010.

\bibitem{EB05}
\BIBentryALTinterwordspacing
S.~C. Edwards and V.~P. Belavkin, ``Optimal quantum filtering and quantum
  feedback control,'' August 2005, {U}niversity of {N}ottingham. [Online].
  Available: \url{http://arxiv.org/pdf/quant-ph/0506018.}
\BIBentrySTDinterwordspacing

\bibitem{BE08}
V.~P. Belavkin and S.~C. Edwards, ``Quantum filtering and optimal control,'' in
  \emph{Quantum Stochastics and Information: Statistics, Filtering and Control
  (University of Nottingham, UK, 15 - 22 July 2006)}, V.~P. Belavkin and
  M.~Guta, Eds.\hskip 1em plus 0.5em minus 0.4em\relax Singapore: World
  Scientific, 2008, pp. 143--205.

\bibitem{WD05}
H.~M. Wiseman and A.~C. Doherty, ``Optimal unravellings for feedback control in
  linear quantum systems,'' \emph{Phys. Rev. Lett.}, vol.~94, pp. 070\,405--1
  -- 070\,405--1, 2005.

\bibitem{Yam06}
N.~Yamamoto, ``Robust observer for uncertain linear quantum systems,''
  \emph{Phys. Rev. A}, vol.~74, pp. 032\,107--1 -- 032\,107--10, 2006.

\bibitem{YNJP08}
N.~Yamamoto, H.~I. Nurdin, M.~R. James, and I.~R. Petersen, ``Avoiding
  entanglement sudden death via measurement feedback control in a quantum
  network,'' \emph{Phys. Rev. A}, vol.~78, p. 042339, 2008.

\bibitem{GGY08}
J.~Gough, R.~Gohm, and M.~Yanagisawa, ``Linear quantum feedback networks,''
  \emph{Phys. Rev. A}, vol.~78, p. 061204, 2008.

\bibitem{BR04}
H.~Bachor and T.~Ralph, \emph{A Guide to Experiments in Quantum Optics},
  2nd~ed.\hskip 1em plus 0.5em minus 0.4em\relax Weinheim, Germany: Wiley-VCH,
  2004.

\bibitem{WaM94}
D.~F. Walls and G.~Milburn, \emph{Quantum Optics}.\hskip 1em plus 0.5em minus
  0.4em\relax Berlin and Heidelberg: Springer-Verlag, 1994.

\bibitem{BvHJ07}
L.~Bouten, R.~{van Handel}, and M.~R. James, ``An introduction to quantum
  filtering,'' \emph{SIAM J. Control Optim.}, vol.~46, pp. 2199--2241, 2007.

\bibitem{GJ07}
J.~Gough and M.~R. James, ``The series product and its application to quantum
  feedforward and feedback networks,'' \emph{{IEEE} Trans. Automat. Contr.},
  vol.~54, no.~11, pp. 2530--2544, 2009.

\bibitem{GJN10}
J.~E. Gough, M.~R. James, and H.~I. Nurdin, ``Squeezing components in linear
  quantum feedback networks,'' \emph{Phys. Rev. A}, vol.~81, pp. 023\,804--1--
  023\,804--15, 2010.

\bibitem{YK03a}
M.~Yanagisawa and H.~Kimura, ``Transfer function approach to quantum
  control-part {I}: Dynamics of quantum feedback systems,'' \emph{IEEE Trans.
  Automatic Control}, vol.~48, no.~12, pp. 2107--2120, 2003.

\bibitem{lin98}
G.~Linblad, ``Brownian motion of harmonic oscillators: Existence of a
  subdynamics,'' \emph{J. Math. Phys.}, vol.~39, no.~5, pp. 2763--2780, 1998.

\end{thebibliography}

\end{document}